\documentclass[twoside]{aiml14}
\usepackage{hyperref}

\usepackage{aiml14macro}

\def\lastname{Kontinen, M\"uller, Schnoor, Vollmer}

\newcommand{\shtitle}{Modal Independence Logic}
\shttltrue

\usepackage{algorithm}
\usepackage[noend]{algpseudocode}
\usepackage{amsmath, amssymb, stmaryrd}
\usepackage{xspace}
\usepackage{xcolor}
\usepackage{needspace}
\usepackage{enumerate}
\usepackage{ifthen}
\usepackage{fancyhdr}
\usepackage{algorithmicx}
\usepackage{tikz}
  \usetikzlibrary{arrows,automata}
\usepackage{listings}
  
  \lstdefinelanguage{pseudo}{
    morekeywords={if,elseif,then,return,end,choose,guess,when,for,foreach,case},
    morekeywords=[3]{false,true,and,or,not},
    morecomment=[l]{//}
  }
\newcommand{\Model}{\mathcal{M}}
\newcommand{\Var}{\textnormal{V}}
\newcommand{\MC}[1]{#1\textnormal{-MC}}
\newcommand{\SAT}[1]{#1\textnormal{-SAT}}
\newcommand{\FO}{\ensuremath{\logicFont{FO}}\xspace}

\newcommand{\D}{\ensuremath{\logicFont{D}}\xspace}
\newcommand{\MIL}{\ensuremath{\logicFont{MIL}}\xspace}
\newcommand{\ML}{\ensuremath{\logicFont{ML}}\xspace}
\newcommand{\MDL}{\ensuremath{\logicFont{MDL}}\xspace}
\newcommand{\EMDL}{\ensuremath{\logicFont{EMDL}}\xspace}
\newcommand{\MILMC}{\MC{\MIL}}
\newcommand{\MILSAT}{\SAT{\MIL}}
\newcommand{\MDLMC}{\MC{\MDL}}
\newcommand{\MDLSAT}{\SAT{\MDL}}
\newcommand{\NP}{\ensuremath{\complexityClassFont{NP}}}
\newcommand{\NEXP}{\complexityClassFont{NEXP}}

\newcommand{\dep}[1][\cdot]{{\textnormal{=}\ifthenelse{\equal{#1}{}}{}{(\nobreak#1\nobreak)}}}

\newcommand{\MILD}{\ensuremath{\ML(D)}\xspace}

\newcommand{\succTeamSet}[1]{R\langle#1\rangle}
\newcommand{\succTeam}[2][R]{#1(#2)}
\newcommand{\md}[1]{\mathop{md}(#1)}
\newcommand{\N}{\mathrm{N}}

\newcommand{\complexityClassFont}[1]{\protect\ensuremath{\mathrm{#1}}\xspace}
\newcommand{\problemFont}[1]{\textsc{#1}\xspace}
\newcommand{\logicFont}[1]{\textsf{\normalshape#1}\xspace}

\newcommand{\set}[1]{\ensuremath\left\{#1\right\}}
\newcommand{\mathtext}[1]{\ensuremath{\mathrm{\text{#1}}}}
\newcommand{\draftnote}[2]{\footnote{\textbf{#1}: #2}}
\newcommand{\hsnote}[1]{\draftnote{HS}{#1}}
\newcommand{\hvnote}[1]{\draftnote{HV}{#1}}

\newcommand{\card}[1]{\left| #1 \right|}
\newcommand\restr[2]{{
  \left.\kern-\nulldelimiterspace
  #1
  \vphantom{\big|}
  \right|_{#2}
  }}

\newlength\problemlength
\newcommand\decisionproblem[3]{%
\begin{nopagebreak}
\begin{list}{}{\labelwidth\problemlength \labelsep.7em \rightmargin1.5em
\leftmargin\problemlength \advance\leftmargin by3em
\parsep0ex \itemsep.2ex plus.1ex}
\item[{\sl Problem:\hfill}] {\problemFont{#1}}
\item[{\sl Input:  \hfill}] #2
\item[{\sl Question: \hfill}] #3
\end{list}
 \end{nopagebreak}
}

\newif\iflongproofs

\begin{document}

\begin{frontmatter}

	\title{Modal Independence Logic}\thanks{The first author was supported by the Academy of Finland grants 264917 and 275241}
	\author{Juha Kontinen\textsuperscript{*}}
	\author{Julian-Steffen M\"uller\textsuperscript{$\dagger$}}
	\author{Henning Schnoor\textsuperscript{$\ddagger$}}
	\author{Heribert Vollmer\textsuperscript{$\dagger$}}
	\address{\textsuperscript{*}University of Helsinki, Department of Mathematics and Statistics, \\P.O. Box 68, 00014 Helsinki, Finland.}
	\address{\textsuperscript{$\dagger$}Leibniz Universit\"at Hannover, Institut f\"ur Theoretische Informatik \\ Appelstr.~4, 30167 Hannover}
	\address{\textsuperscript{$\ddagger$}Institut f\"ur Informatik, Christian-Albrechts-Universit\"{a}t zu Kiel\\ 24098 Kiel}

\begin{abstract}
This paper introduces modal independence logic \MIL, a modal logic that can explicitly talk about independence among propositional variables. Formulas of \MIL are not evaluated in worlds but in sets of worlds, so called \emph{teams}. In this vein, \MIL can be seen as a variant of V\"a\"an\"anen's modal dependence logic \MDL. We show that \MIL embeds \MDL and is strictly more expressive. However, on singleton teams, \MIL is shown to be not more expressive than usual modal logic, but \MIL is exponentially more succinct.
Making use of a new form of bisimulation, we extend these expressivity results to modal logics extended by various generalized dependence atoms. 
We demonstrate the expressive power of \MIL by giving a specification of the anonymity requirement of the \emph{dining cryptographers} protocol in \MIL.  
We also study complexity issues of \MIL and show that, though it is more expressive, its satisfiability and model checking problem have the same complexity as for \MDL. 
\end{abstract}


\end{frontmatter}

\section{Introduction}

The concept of independence is ubiquitous in many scientific disciplines such as experimental physics, social choice theory, computer science, and cryptography.  Dependence logic \D, introduced by Jouko V\"a\"an\"anen in \cite{vaananen07}, is a new logical framework in which various notions of dependence and independence can be formalized and studied. It extends first-order logic by so called dependence atoms
$$\dep[x_{1},\dots,x_{n-1},x_{n}],$$
expressing that the value of the variable $x_{n}$ depends (only) on the values of $x_{1},\dots,x_{n-1}$, in other words, that $x_{n}$ is functionally dependent of $x_{1},\dots,x_{n-1}$.
Of course, such a dependency does not make sense when talking about single assignments; therefore dependence logic formulas are evaluated for so called \emph{teams}, i.\,e., sets of assignments. A team can for example be a relational database table, a collection of plays of a game, or a set of agents with features. It is this \emph{team semantics}, together with dependence atoms, that gives dependence logic its expressive power: it is known that \D is as expressive as $\Sigma_{1}^1$, that is, the properties of finite structures that can be expressed in dependence logic are exactly the $\NP$-properties.

In a slightly later paper V\"a\"an\"anen \cite{va09} introduced dependence atoms into (propositional) modal logic. Here, teams are sets of worlds, and a dependence atom $\dep[p_{1},\dots,p_{n-1},p_{n}]$ holds in a team $T$ if there is Boolean function that determines the value of $p_{n}$ from those of $p_{1},\dots,p_{n-1}$ in all worlds in $T$. The so obtained modal dependence logic \MDL was studied from the point of view of expressivity and complexity in \cite{se09}.

In this article we introduce a novel modal variant of dependence logic called \emph{modal independence logic}, $\MIL$, extending the formulas of  modal logic $\ML$ by  so-called \emph{independence} atoms
		$$(p_1, \dots, p_\ell) \bot_{(r_1, \dots ,r_m)} (q_1, \dots, q_n),$$
the meaning of which  is that the propositional sequences $\pol{p}$ and $\pol{q}$ are independent of each other for any fixed value of  $\pol{r}$. 
Modal independence logic thus has its roots in modal dependence logic $\MDL$ \cite{va09} and first-order independence logic \cite{DBLP:journals/sLogica/GradelV13}.
In modal independence logic, dependencies between propositions can be expressed, and thus, analogously to the first-order case,  $\MDL$ can be embedded as a sublogic into $\MIL$, and 
it is easy to see that
\MIL is strictly more expressive than \MDL. 

The aim of this paper is to initiate a study of the expressiveness and the computational complexity of modal independence logic. 
For this end, we first study the computational complexity of the satisfiability and the model checking problem for $\MIL$. We show that, though \MIL is more expressive than \MDL, the complexity of these decision problems stays the same, i.\,e., the satisfiability problem is complete for nondeterministic exponential time ($\NEXP$-complete, \cite{se09}) and the model checking problem is $\NP$-complete \cite{eblo12}. In order to settle the complexity of satisfiability for \MIL, we give a translation of \MIL-formulas to existential second-order logic formulas the first-order part of which is in the  G\"odel-Kalm\'ar-Sch\"utte prefix class. Our result then follows from the classical result that the satisfiability problem for this prefix class is  $\NEXP$-complete \cite{bogrgu01}. We will also show that the same  upper bound on satisfiability can be obtained for a whole range of variants of \MIL via the notion of a generalized  (modal) dependence atom (a notion introduced in the first-order framework in \cite{ku13}).
	
		The expressive power  of $\MDL$ was first studied by Sevenster \cite{se09}, where he showed that $\MDL$ is equivalent to $\ML$ on singleton teams. In this paper we prove a general result showing that  \MIL, and in fact  any variant of it whose generalized dependence atoms are   $\FO$-definable, 
is bound to be equivalent to $\ML$ over singleton teams. Interestingly, it was recently shown in \cite{ehmmvh13}  that a so-called extended modal dependence logic $\EMDL$ is strictly more expressive than $\MDL$ even on singletons.

To demonstrate the potential applications of \MIL, we consider  
the \emph{dining cryptographers} protocol~\cite{Chaum-DINING-CRYPTOGRAPHERS-JCRYPT-1988}, a classic example for anonymous broadcast which is used as a benchmark protocol in model checking of security protocols~\cite{AlBatainehvdMeyden-ABSTRACTION-MODEL-CHECKING-DINING-CRYPTOGRAPHERS-TARK-2011}. We show how the anonymity requirement of the protocol can be formalized in modal independence logic, where---unlike in the usual approaches using epistemic logic---we do not need to use the Kripke model's accessibility relation to encode knowledge, but to express the ``possible future'' relation of branching-time models. In addition to demonstrating \MIL's expressivity, we also derive a succinctness result from our modeling of the dining cryptographers: While \MIL and \ML are equally expressive on singletons, \MIL is exponentially more succinct.

		\section{Modal Independence Logic}

		\begin{definition}
			The syntax of \emph{modal logic} \ML is inductively defined by the following grammar in extended Backus Naur form:
			$$\phi ::= p \mid \overline{p}
			\mid \phi \wedge \phi \mid \phi \vee \phi 
			\mid \Diamond\phi \mid \Box\phi.$$

			The syntax of \emph{modal dependence logic} \MDL is defined by 
			$$\phi ::= p \mid \overline{p} \mid \dep[\pol{q},p]
			\mid \phi \wedge \phi \mid \phi \vee \phi 
			\mid \Diamond\phi \mid \Box\phi,$$
			where $p$ is a propositional variable and $\pol{q}$ a sequence of propositional variables. 

			The syntax of \emph{modal independence logic} \MIL is defined by 
			$$\phi ::= p \mid \overline{p} \mid \pol{p} \bot_{\pol{r}} \pol{q}
			\mid \phi \wedge \phi \mid \phi \vee \phi 
			\mid \Diamond\phi \mid \Box\phi,$$
			where $p$ is a propositional variable and $\pol{p}, \pol{r}, \pol{q}$ are sequences of propositional variables. The sequence $\pol{r}$ may be empty.
		\end{definition}

	A Kripke structure is a tuple $\Model = (W,R, \pi)$, where $W$ is a non-empty set of worlds, $R$ is a binary relation over $W$ and $\pi\colon W \rightarrow \mathcal{P}(\Var)$ is a labeling function. 
	As usual, in a Kripke structure $\Model$ the set of all successors of $T\subseteq W$ is defined as $\succTeam{T} = \{s \in W \mid \exists s' \in T:  (s',s) \in R\}$.
Furthermore we define $\succTeamSet{T} = \{T'\subseteq \succTeam{T} \mid \forall s\in T ~ \exists s' \in T' : (s,s') \in R\}$, the set of legal successor teams.

	\begin{definition}\label{def:agree}
		Let $\pol{p} = (p_1, \dots, p_n)$ be a sequence of variables and $w,w'$ be worlds of a Kripke model $\Model= (W,R, \pi)$. Then $w$ and $w'$ are equivalent under $\pi$ over $\pol{p}$, denoted by $w \equiv_{\pi, \pol{p}} w'$, if the following holds:
		$$\pi(w) \cap \{p_1, \dots, p_n\} = \pi(w') \cap \{p_1, \dots, p_n\}.$$
	\end{definition}

	\begin{definition} {(Semantics of \ML, \MDL, and \MIL)}
		Let $\Model= (W,R, \pi)$ be a Kripke structure, $T$ be a team over $\Model$ and $\phi$ be a $\MIL$-formula. The semantic evaluation (denoted as $\Model,T \models \phi$) is defined inductively as follows.
		$$\begin{array}{l@{\quad}l@{\quad}l}
			\Model, T \models p & \Leftrightarrow & \forall w \in T\colon p \in \pi(w)\\
			\Model, T \models \overline p & \Leftrightarrow & \forall w \in T\colon p \not\in \pi(w)\\
			\Model, T \models \phi_1 \wedge \phi_2 & \Leftrightarrow & \Model, T \models \phi_1 \text{ and } \Model, T \models \phi_2\\
			\Model, T \models \phi_1 \vee \phi_2 & \Leftrightarrow & \exists T_1,T_2\colon T_1 \cup T_2 = T, \Model, T_1 \models \phi_1 \text{ and } \Model, T_2 \models \phi_2\\
			\Model, T \models \Diamond \phi & \Leftrightarrow & \exists T' \in\succTeamSet{T}\colon \Model, T' \models \phi\\
			\Model, T \models \Box \phi & \Leftrightarrow & \Model, \succTeam{T} \models \phi\\
			\Model, T \models \dep[\pol{q},p] & \Leftrightarrow & \forall w,w' \in T \colon w \equiv_{\pi,\pol{q}} w' \text { implies } w \equiv_{\pi, p} w'\\
			\Model, T \models \pol{p}_1  \bot_{\pol{q}}\ \pol{p}_2 & \Leftrightarrow & \forall w,w' \in T \colon w \equiv_{\pi,\pol{q}} w' \text { implies } \exists w'' \in T \colon \\
			&& w'' \equiv_{\pi, \pol{p}_1} w \text{ and } w'' \equiv_{\pi, \pol{p}_2} w' \text{ and } w'' \equiv_{\pi, \pol{q}} w
		\end{array}$$
	\end{definition}	
	
Note that for modal logic formulas $\phi$ we have $\Model,\{w\}\models\phi$ iff $\Model,w\models\phi$ (where in the latter case, $\models$ is defined as in any textbook for usual modal logic). In fact it is easy to see that without dependence or independence atom, our logic has the so called \emph{flatness property}, stating that team semantics and usual semantics essentially do not make a difference: 

\begin{lemma}
For every \ML-formula $\phi$ and all models $M$ and teams $T$, $M,T\models \phi$  iff  $M,w \models \phi$ for all $w\in T$. 
\end{lemma}

Team semantics and independence atoms together will lead to a richer expressive power, as we will prove in Sect.~\ref{sect:expressiveness}. However, we will also show that over teams $T$ consisting of one world only, \ML and \MIL have the same expressive power.
		
\begin{definition}
Formulas $\varphi$ and $\varphi'$ are \emph{equivalent on singletons}, if for every model $M$ and every $w\in M$, we have  
$M,\set{w}\models\varphi$ if and only if $M,\{w\}\models\varphi'$.
\end{definition}

	\section{Complexity Results}\label{sect:complexity}

	In this section we will study the computational complexity of the model checking and the satisfiabilty problem for $\MIL$. In~\cite{DBLP:journals/sLogica/GradelV13}, it was observed that in first-order logic, $\dep[\pol{p},\pol{q}]$ is equivalent to $\pol{q}\ \bot_{\pol{p}}\pol{q}$. This observation clearly carries over to \MIL, and hence in particular shows that \MIL is a generalization of \MDL.

	\begin{lemma}\label{lemma:dep_simulation}
	 Let $\Model$ be a model and $T$ a team over $\Model$, let $\pol{p}$ and $\pol{q}$ be sets of variables. Then $\Model,T\models\dep[\pol{p},\pol{q}]$ if and only if $\Model,T\models\pol{q}\ \bot_{\pol{p}}\pol{q}$.
	\end{lemma}


We now define the two decision problems whose complexity we wish to study, namely the model checking and the satisfiability problem for modal independence logic.
\decisionproblem{$\MILSAT$}
	{\MIL formula $\phi$}
	{Does there exists a Kripke model $\Model$ and a team $T$ with $\Model, T \models \phi$?}
\decisionproblem{$\MILMC$}
	{Kripke model $\Model$, team $T$ and \MIL formula $\phi$}
	{$\Model, T \models \phi$?}
	
The corresponding problems for modal dependence logic are denoted by $\MDLSAT$ and $\MDLMC$. 

It is easy to see that model checking for \MIL is not more difficult than model checking for \MDL, namely \NP-complete. 
	
	\begin{theorem}\label{theorem:milmc np complete}
		$\MILMC$ is $\NP$-complete.
	\end{theorem}
	
\begin{proof}
%
	The lower bound follows immediately from Lemma~\ref{lemma:dep_simulation} and \NP-completeness of $\MDLMC$ \cite{eblo12}. The upper bound follows from a simple extension of the well-known model checking algorithm for modal logic, see Algorithm~\ref{alg:milmc}.

	\begin{algorithm} 
			\caption{$\NP$ algorithm for $\MILMC$}
			\label{alg:milmc}
		\begin{algorithmic}[1]
			\Function{milmc}{$\Model, T, \phi$} 
			\If{$\phi = \Box \psi$}
				\State \textbf{return} \Call{milmc}{$M,\psi,\succTeam{T}$}
			\EndIf
			\If{$\phi = \Diamond \psi$}
				\State \textbf{existentially guess} $T' \in \succTeamSet{T}$
				\State \textbf{return} \Call{milmc}{$M,\psi,T'$}
			\ElsIf{$\phi = \psi_1 \wedge \psi_2$}
				\State \textbf{return} \Call{milmc}{$M,\psi_1,T$} \textbf{ and } \Call{milmc}{$M,\psi_2,T$}
			\ElsIf{$\phi = \psi_1 \vee \psi_2$}
				\State \textbf{existentially guess} $T_1 \cup T_2 = T$
				\State \textbf{return} \Call{milmc}{$M,\psi_1,T_1$} \textbf{ and } \Call{milmc}{$M,\psi_2,T_2$}
			\ElsIf{$\phi = p$}
				\For{$s \in T$}
					\If{$p \not\in \pi(s)$}
						\State \textbf{return} \textbf{false}	
					\EndIf
				\EndFor	
				\State \textbf{return} \textbf{true}
			\ElsIf{$\phi = \overline{p}$}
				\For{$s \in T$}
					\If{$p \in \pi(s)$}
						\State \textbf{return} \textbf{false}	
					\EndIf
				\EndFor	
				\State \textbf{return} \textbf{true}
			\ElsIf{$\phi = \pol{p}\ \bot_{\pol{r}}\ \pol{q}$}
				\For{$s \in T$}
					\For{$s' \in T$}
						\If{$\pi(s') \cap \pol{r} = \pi(s'') \cap \pol{r}$}
							\State found $\leftarrow$ \textbf{false}
							\For{$s'' \in T$} 
								\State agreeP $\leftarrow \pi(s'') \cap \pol{p} = \pi(s) \cap \pol{p}$ 
						
								\State agreeQ $\leftarrow \pi(s'') \cap \pol{q} = \pi(s') \cap \pol{q}$ 
								
								\State agreeR $\leftarrow \pi(s'') \cap \pol{r} = \pi(s) \cap \pol{r}$ 

								\If{\textnormal{agreeP} \textbf{ and } \textnormal{agreeQ} \textbf{ and } \textnormal{agreeR}}
									\State found $\leftarrow$ \textbf{true}
								\EndIf
							\EndFor
							\If{$\textbf{not }\textnormal{found}$}
								\State \textbf{return } \textbf{false}
							\EndIf
						\EndIf
					\EndFor
				\EndFor
				\State \textbf{return} \textbf{true}
			\EndIf
			\EndFunction
		\end{algorithmic}
	\end{algorithm}
\end{proof}

Next we will consider the complexity of the satisfiability problem $\MILSAT$ for	modal independence logic. From Lemma~\ref{lemma:dep_simulation} and the hardness of $\MDLSAT$ for nondeterministic exponential time~\cite{se09} we immediately obtain the following lower bound:

	\begin{lemma}
		$\MILSAT$ is $\NEXP$-hard.
	\end{lemma}

In order to  show containment in  $\NEXP$, we need to recall the following classical result. Recall that the so-called G\"odel-Kalm\'ar-Sch\"utte  prefix class $[\exists ^*\forall^2\exists ^*,all]$ contains sentences of FO, in a relational vocabulary without equality,  which are in prenex normal form and have a quantifier prefix of the form $\exists ^*\forall^2\exists ^*$.

\begin{proposition}[\cite{bogrgu01}] \label{prefix}
Satisfiability of formulas in prefix class  $[\exists ^*\forall^2\exists ^*,all]$ can be decided in $\mathrm{NTIME}(2^{O(n/\log n)})$.
\end{proposition}

Next we will show that $\MILSAT\in \NEXP$ with the help of Proposition~\ref{prefix}. We will first define a variant of the standard translation of ML into FO in the case of $\MIL$.  For a Kripke structure $(W,R,\pi)$, and a team $T \subseteq W$, we denote by $(W, \{A_i\}_i, R, T)$  the first-order structure of vocabulary $\{R,T \}\cup \{A_i\}_{i\in \N}$ encoding  $(W,R,\pi)$ in the obvious way.

\begin{lemma}\label{lemmaeso} For any formula $\phi\in \MIL$ there is a sentence $\phi^*$ of existential (monadic) second-order logic of the form
\begin{equation}\label{translation}
 \exists Y_1\ldots Y_m\forall xy\exists z_1\ldots z_k\theta  ,
 \end{equation}
where $\theta$ is quantifier-free, and  such that for all  $(W,R,\pi)$ and $T$ it holds 
 \[ (W,R,\pi),T\models \phi \Leftrightarrow   (W, \{A_i\}_{i\in \N}, R, T)  \models \phi^*.  \]
\end{lemma}
\begin{proof}
We first define an auxiliary translation  $\phi\mapsto \phi'$ for which correctness is obvious and then indicate how to go from $\phi'$ to $\phi^*$.
\begin{enumerate}
	\item Suppose $\phi$ is $p_i$. Then $\phi'$ is defined as 
	\[  \phi':=    \forall x (T(x) \rightarrow A_{p_i}(x)).      \] 
	
	\item Suppose $\phi$ is $\overline p_i$. Then $\phi'$ is defined as 
	\[  \phi':=    \forall x (T(x) \rightarrow \neg A_{p_i}(x)).      \] 
	
	\item\label{vee} Suppose $\phi$ is $\psi_1 \vee \psi_2$. Then $\phi'$ is defined as 
	\[  \phi':=  \exists Y_1Y_2(\psi'_1(T/Y_1) \wedge \psi'_2(T/Y_2)) \wedge \forall x (T(x) \rightarrow (Y_1(x) \vee Y_2(x))) \] 
	\item Suppose $\phi$ is $\psi_1 \wedge \psi_2$. Then $\phi'$ is defined as 
	\[  \phi':=  \psi'_1 \wedge \psi'_2. \] 
\item Suppose $\phi$ is $\Box\psi$. Then $\phi'$  is defined as 
	\[  \phi':=  \exists Y (\psi'(T/Y) \wedge \forall x  \forall y ((T(x) \wedge E(x,y)) \leftrightarrow  Y(y)))
\] 
	\item Suppose $\phi$ is $\Diamond\psi$. Then $\phi'$  is defined as 
	\[  \phi':=  \exists Y (\psi'(T/Y) \wedge \forall x (T(x) \rightarrow \exists z (Y(z) \wedge E(x,z))) )\]
	
	\item Suppose $\phi$ is $\pol{p}_1  \bot_{\pol{p}_2} \pol{p}_3$. Then $\phi'$  is defined as 
	\begin{eqnarray*}
   \phi':=\forall x \forall y ((T(x) \wedge T(y) \wedge EQ_{\pol{p}_2}(x,y))\rightarrow  && \\
	\exists z(T(z) \wedge  EQ_{\pol{p}_2}(x,z) \wedge && EQ_{\pol{p}_1}(x,z) \wedge     
EQ_{\pol{p}_3}(y,z))),
\end{eqnarray*}
where $EQ_{\pol{p}_i}(v,w)$ is a shorthand for the formula
\[\bigwedge_{p \in \pol{p}_i}  A_{p}(v) \leftrightarrow A_{p}(w) .   \]
\end{enumerate}
It remains  to define the translation $\phi \mapsto \phi^*$. This translation is defined by modifying  the above clauses by essentially moving all quantifiers to the left of the formula, and by possibly renaming some of the bound variables $Y_i$ and $z_i$.  We will indicate these modifications by considering the case of disjunction. The other cases are analogous. Assume that  $\psi^*_1 $ and  $\psi^*_2$ are defined already:
\[  \psi_i^* = \exists \bar{Y}_i\forall xy\exists \pol{z}_i \theta_i,  \]
where $\theta_i$ is  quantifier free, and  $\psi^*_i\equiv \psi'_i $. By renaming of bound variables, we may assume that 
$\bar{Y}_2=Y_3\dots Y_k$, and $\bar{Y}_1=Y_{k+1}\dots Y_m$, and that $\pol{z}_1$ and $\pol{z}_2$ do not have any common variables either. Then  $(\psi_1\vee \psi_2)^*$  is defined by replacing $\psi'_i $ by $\psi^*_i $ in the definition of $(\psi_1\vee \psi_2)'$ (see clause  \ref{vee}), and by extending the scopes of the quantifiers:
\begin{eqnarray*}
 (\psi_1\vee \psi_2)^*:= \exists Y_1\ldots Y_m \forall xy\exists \pol{z}_2 \pol{z}_1 ( (\theta_1(T/Y_1) \wedge && \theta_2(T/Y_2)) \wedge   \\
  (T(x) \rightarrow && (Y_1(x) \vee Y_2(x))).  
  \end{eqnarray*}
\end{proof}

		\begin{theorem}\label{theorem:milsat in nexp}
		$\MILSAT$ is in $\NEXP$.
	\end{theorem}
\begin{proof}	
Let $\phi\in \MIL$. Then $\phi$ is satisfiable if and only if $\phi^*$ is satisfiable. This follows from the previous theorem and the fact that there is a 1-1 correspondence with Kripke structures $(W,R,\pi)$, and teams $T $ 
for $\phi$ and  $\{R,T \}\cup \{A_i\}_{1 \le i\le n}$-structures $(W, \{A_i\}_{1 \le i\le n},R, T)$, where $n$ is large enough such that all $p_i$ appearing in $\phi$ satisfy $i\le n$. 

Recall  now that $\phi^*$ has the form \eqref{translation}, hence it is satisfiable if and only if the first-order sentence
\begin{equation}\label{FOpart}
\forall xy\exists z_1\ldots z_k\theta 
 \end{equation}
of vocabulary $\{Y_1,\ldots,Y_m\}\cup \{R,T \}\cup \{A_i\}_{1 \le i\le n}$ is satisfiable. The sentence \eqref{FOpart} is contained in prefix class $[\exists ^*\forall^2\exists ^*,all]$, hence  the satisfiability of it, and also  $\phi^*$, can be decided in time $\mathrm{NTIME}(2^{O(|\phi^*|)})$. The claim now follows from the fact  the mapping $\phi \mapsto \phi^*$ can be computed in time polynomial in $|\phi |$.
\end{proof}

\begin{corollary}
$\MILSAT$ is \NEXP-complete.  
\end{corollary}

It is interesting to note that Theorem~\ref{theorem:milsat in nexp} and Lemma~\ref{lemma:dep_simulation} directly imply the result of Sevenster \cite{se09} that  $\MDLSAT$ is contained in $\NEXP$. On the other hand, it seems that the original argument of Sevenster does not immediately generalize to $\MIL$.

\begin{corollary}
$\MDLSAT$ is \NEXP-complete.  
\end{corollary}

\section{Generalized Dependency Notions}\label{sect:generalized dependence}

$\MIL$ can be seen as an extension of modal logic with team semantics by the independence atom---let us denote such an extension by $\ML(\bot)$. Similarly, we can extend modal logic with other atoms, so-called \emph{generalized dependence atoms}, which we define now.

\begin{definition}
	Let $\Model=(W,R,\pi)$ be a Kripke model and $T=(w_1, \dots, w_m)$ be a team over $\Model$. Then for any propositional variable $p$, $T(p)$ is defined as the tuple $(s_1, \dots, s_m)$, where $s_i$, for $1 \leq i \leq m$ is defined as:
		$$s_i = 	\begin{cases}
					1	&	w_i \in \pi(p)\\
					0	&	\text{, otherwise}	
				\end{cases}.$$
\end{definition}

For a set of propositions $\pol{q} = (q_1, \dots, q_k)$, we define $T(\pol{q})$ analogously as $(T(q_1), \dots, T(q_k))$. 

Similar to Kuusisto's \cite{ku13} definition of generalized first order dependence atoms we give a definition of generalized modal dependence atoms. In the following, a set of matrices $D$ is invariant under permutations of rows, if for every matrix $M\in D$, if $M'$ is obtained from $M$ by permuting $M$'s rows, then $M'$ is an element of $D$ as well.

\begin{definition}
Let $D$ be a set of Boolean $n$-column matrices that is invariant under permutation of rows.
The semantics of the \emph{generalized dependence atom defined by $D$} is given as follows:

	Let $\Model$ be a Kripke model, $T$ be a team over $\Model$ and $p_1, \dots, p_n$ atomic propositions. Then
	$$\Model, T \models D(p_1, \dots, p_n) \;\Longleftrightarrow\; \langle T(p_1), \dots, T(p_n) \rangle \in D.$$
The \emph{width} of $D$ is defined to be $n$.

Note that for simplicity we do not distinguish in notation between the logical atom $D$ and the set $D$ of Boolean matrices. 
\end{definition}

The Boolean matrix $\langle T(p_1), \dots, T(p_n) \rangle$ contains one column for each of the variables $p_1,\dots,p_n$; each row of the matrix corresponds to one world from $T$. The entry for variable $p_i$ and world $w\in T$ is $1$ if and only if the variable $p_i$ is satisfied in the world $w$. We require that $D$ is invariant under permutation of rows in order to ensure that whether $\Model,T\models D(p_1,\dots,p_n)$ holds does not depend on the ordering of the worlds in $T$ that is used in computing the tuple $T(p)$. 

In the following we will mainly be interested in generalized dependence atom \emph{definable} by first order formulae. For this purpose let $D$ be an atom of width $n$ as above, and 
$\phi$ be a first order sentence over signature $\langle A_{1}, \dots, A_{n}\rangle$.
Then $\phi$ \emph{defines} $D$ if for all Kripke models $\Model = (W,R,\pi)$ and teams $T$ over $\Model$,
		$$\Model, T \models D(p_1, \dots, p_n) \;\Longleftrightarrow\; \mathcal{A} \models \phi,$$
	where $\mathcal{A}$ is the first order structure  with universe $T$ and relations $A_{i}^{\mathcal{A}}$ for $1 \leq i \leq n$, where for all $w\in T$, $w \in A_{i}^{\mathcal{A}} \Leftrightarrow p_i \in \pi(w)$.

We say that a generalized dependence atom $D$ is \emph{\FO-definable} if there exists a formula $\phi$ defining $D$ as above. 
Strictly speaking, the dependence atoms considered in the literature are \emph{families} of dependence atoms for different width, e.\,g., the simple dependence $\dep[p_{1},\dots,p_{n}]$ is defined for arbitrary values of $n$. Let us say that such a family is (P-uniformly) \FO-definable if there exists a family of defining first order formulae $\phi_{n}$ such that $\phi_{n}$ defines the atom of width $n$ and the mapping $1^n\mapsto\langle\phi_n\rangle$ is computable in polynomial time; that is, an encoding of formula $\phi_n$ is computable in time polynomial in $n$. Note that in particular this implies that $|\phi_{n}| = p(n)$ for some polynomial $p$.


As examples let us show how to define some well-studied generalized dependence atoms as follows.

$$\begin{array}{llll}
	\dep[\pol{p},q] & \Leftrightarrow & \forall w \forall w' & ((\bigwedge_{1 \leq i \leq n} A_{p_i}(w) \leftrightarrow A_{p_i}(w')) \rightarrow (A_{q}(w) \leftrightarrow A_{q}(w')))\\
	\pol{p} \subseteq \pol{q} &\Leftrightarrow& \forall w \exists w' & (\bigwedge_{1 \leq i \leq n} A_{p_i}(w) \leftrightarrow A_{q_i}(w'))\\
	\pol{p} \,\mid\,\! \pol{q}& \Leftrightarrow &\forall w \forall w' & (\bigvee_{1 \leq i \leq n} A_{p_i}(w) \leftrightarrow \neg A_{q_i}(w'))\\
\end{array}$$

The latter two so-called \emph{inclusion} and \emph{exclusion} atoms were introduced by Galliani in \cite{Galliani:2011}.
In particular all above atoms are \FO-definable. The independence atom $\pol{p}_1\bot_{\pol{q}}\pol{p}_2$ is also \FO-definable in the obvious way\footnote{Since the \FO-formula $\phi$ may only depend on the width, we restrict ourselves to occurrences of $\pol{p}_1\bot_{\pol{q}}\pol{p}_2$ where $\card{\pol{p}_1}=\card{\pol{p}_2}=\card{\pol{q}}$, if these sets are nonempty, which we can always assume without loss of generality by repeating variable occurrences, the case that one of these sets is empty can then be encoded into widths that are not multiples of $3$ in a straightforward manner.}.

We use \MILD to denote the logic obtained from \MIL by replacing the independence atom with the generalized dependence atom $D$.
Our complexity upper bounds from Section~\ref{sect:complexity} can be generalized to 
arbitrary dependence atoms by adding processing rules which verify the corresponding first order formulae. For model checking, we simply use the fact that first-oder formulas can be verified in polynomial time and obtain the following corollary:

\begin{corollary}\label{corollary:generalized atoms model checking}
	Let $D$ be an \FO-definable generalized dependence atom. 
	Then $\MC{\MILD}$ is in $\NP$.
\end{corollary}

For the satisfiability problem, we generalize the proof of Theorem~\ref{lemmaeso} in case~(vii) to arbitrary dependence atoms which are definable by a $[\exists^*\forall^2\exists^*]$ formula.

\begin{corollary}\label{corollary:generalized atoms satisfiability}
	Let ${D}$ be a generalized dependence atom that is \FO-definable by a (family of) first order formula(e) in the prefix class $[\exists^*\forall^2\exists^*]$. 
	Then $\SAT{\ML({D})}$ is in $\NEXP$.
\end{corollary}

\section{Example: The Dining Cryptographers}\label{sect:dining cryptographers}

\newcommand{\bit}{\mathit{bit}}
\newcommand{\announce}{\mathit{announce}}
The dining cryptographers~\cite{Chaum-DINING-CRYPTOGRAPHERS-JCRYPT-1988}, a standard example for anonymous broadcast, is the following problem: A group of cryptographers $\set{c_0,\dots,c_{n-1}}$ with $n\ge 3$ sit in a restaurant, where $c_i$ sits between $c_{i-1}$ and $c_{i+1}$. (Indices of the cryptographers are always modulo $n$, and $i$ always ranges over $0,\dots,n-1$). After dinner, it turns out that someone already paid. There are only two possibilities: Either one of the cryptographers secretly paid, or the NSA did. Naturally, they want to know which of these is the case, but without revealing the paying cryptographer if one of them paid. They use the following protocol:

\begin{itemize}
 \item For each $i$, let $p_i$ be $1$ iff $c_i$ paid. Each $c_i$ knows the value of $p_i$, but not of $p_j$ for $j\neq i$. There is at most one $i$ with $p_i=1$. The protocol computes the value $p_0\oplus p_1\oplus\dots\oplus p_{n-1}$, which is the same as $p_0\vee p_1\vee\dots\vee p_{n-1}$.
 \item Each adjacent pair $\set{c_i,c_{i+1}}$ computes a random bit $\bit_{\set{i,i+1}}$.
 \item Each $c_i$ publicly announces the value $\announce_i=p_i\oplus\bit_{\set{i,i-1}}\oplus\bit_{\set{i,i+1}}$.
 \item Then, $p_0\oplus p_1\oplus\dots\oplus p_{n-1}=\announce_0\oplus \announce_1\oplus\dots\oplus \announce_{n-1}$.
\end{itemize}

The protocol clearly computes the correct answer, the interesting aspect is the \emph{anonymity requirement}: No cryptographer $c_i$ should learn anything about the values $p_j$ for $j\neq i$ except for what follows from the values $p_i$ or the result (if $c_i$ or the NSA paid then $c_j$ did not). The protocol models anonymous broadcast, since the message ``$1$'' is, if sent, received by all cryptographers, but the sender remains anonymous. We formalize this using modal independence logic. We start by capturing the protocol in the following Kripke model:

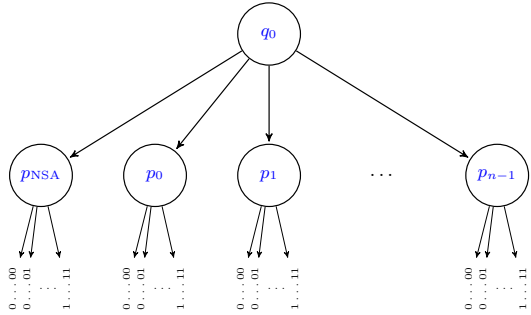
\begin{center}
\scalebox{0.75}{
\begin{tikzpicture}[->,>=stealth',shorten >=1pt,auto,node distance=1.75cm, semithick]
    \tikzstyle{every state}=[fill=white,draw=black,text=blue]

    \node[state,minimum size=1.1cm] at (0,0.5)   (A) {\begin{small}$q_0$\end{small}};
    \node[state,minimum size=1.1cm] at (-4,-2) (B) {\begin{small}$p_{\mathtext{NSA}}$\end{small}};
    \node[state,minimum size=1.1cm] at (-2,-2) (C) {\begin{small}$p_0$\end{small}};
    \node[state,minimum size=1.1cm] at (0,-2) (D) {\begin{small}$p_1$\end{small}};
    \node at (2,-2) {$\dots$};
    \node[state,minimum size=1.1cm] at (4,-2) (E) {\begin{small}$p_{n-1}$\end{small}};
    
    \draw [->] (A) edge (B);
    \draw [->] (A) edge (C);
    \draw [->] (A) edge (D);
    \draw [->] (A) edge (E);
    
    \node at (-4.5,-4.00)  (B1) {\begin{tiny}\rotatebox{90}{$0\dots00$}\end{tiny}};
    \node at (-4.25,-4.00) (B2) {\begin{tiny}\rotatebox{90}{$0\dots01$}\end{tiny}};
    \node at (-3.85,-4.00) (B3) {\begin{tiny}$\dots$\end{tiny}};
    \node at (-3.55,-4.00) (B4) {\begin{tiny}\rotatebox{90}{$1\dots11$}\end{tiny}};
    
    \node at (-2.5,-4.00)  (C1) {\begin{tiny}\rotatebox{90}{$0\dots00$}\end{tiny}};
    \node at (-2.25,-4.00) (C2) {\begin{tiny}\rotatebox{90}{$0\dots01$}\end{tiny}};
    \node at (-1.85,-4.00) (C3) {\begin{tiny}$\dots$\end{tiny}};
    \node at (-1.55,-4.00) (C4) {\begin{tiny}\rotatebox{90}{$1\dots11$}\end{tiny}};
    
    \node at (-0.5,-4.00)  (D1) {\begin{tiny}\rotatebox{90}{$0\dots00$}\end{tiny}};
    \node at (-0.25,-4.00) (D2) {\begin{tiny}\rotatebox{90}{$0\dots01$}\end{tiny}};
    \node at ( 0.15,-4.00) (D3) {\begin{tiny}$\dots$\end{tiny}};
    \node at ( 0.45,-4.00) (D4) {\begin{tiny}\rotatebox{90}{$1\dots11$}\end{tiny}};

    \node at ( 3.5,-4.00)  (E1) {\begin{tiny}\rotatebox{90}{$0\dots00$}\end{tiny}};
    \node at ( 3.75,-4.00) (E2) {\begin{tiny}\rotatebox{90}{$0\dots01$}\end{tiny}};
    \node at ( 4.15,-4.00) (E3) {\begin{tiny}$\dots$\end{tiny}};
    \node at ( 4.45,-4.00) (E4) {\begin{tiny}\rotatebox{90}{$1\dots11$}\end{tiny}};
    
    \draw [very thin,->] (B) edge (B1);
    \draw [very thin,->] (B) edge (B2);
    \draw [very thin,->] (B) edge (B4);
    
    \draw [very thin,->] (C) edge (C1);
    \draw [very thin,->] (C) edge (C2);
    \draw [very thin,->] (C) edge (C4);

    \draw [very thin,->] (D) edge (D1);
    \draw [very thin,->] (D) edge (D2);
    \draw [very thin,->] (D) edge (D4);

    \draw [very thin,->] (E) edge (E1);
    \draw [very thin,->] (E) edge (E2);
    \draw [very thin,->] (E) edge (E4);
\end{tikzpicture}}
\end{center}

The protocol starts in $q_0$, the model then branches into states $p_{\mathtext{NSA}}$, $p_0$, \dots, $p_{n-1}$, depending on whether the NSA or some $c_i$ paid. Each of these states has $2^n$ successor states, for the $2^n$ possible random bit values, these states are \emph{final}. The relation $R$ is as indicated. We use the following variables:

\begin{itemize}
 \item $p_{\mathtext{NSA}}$ and $p_i$ are true if the NSA, resp.~cryptographer $c_i$ paid, i.e., in the states denoted with the same name as the variable and in their successors.
 \item each of the $n$ variables $\bit_{\set{i,i+1}}$ is true in the states where the bit shared between $c_i$ and $c_{i+1}$ is $1$.
 \item each $\announce_i$ is true in all final states which satisfy $p_i\oplus\bit_{\set{i,i-1}}\oplus\bit_{\set{i,i+1}}$ (this encodes that the cryptographers follow the protocol).
\end{itemize}

For each $c_i$, we define the set $\pol{k}_i$ of the variables whose values $c_i$ knows after the protocol run as $\pol{k}_i:=\set{p_i,\bit_{\set{i,i-1}},\bit_{\set{i,i+1}}}\cup\set{\announce_j\ \vert\ j\neq i}$. Clearly, $c_i$ also knows the value $\announce_i$, but since this can be computed from $p_i$, $\bit_{\set{i,i-1}}$ and $\bit_{\set{i,i+1}}$, we omit it from $\pol{k}_i$. 

The formula expressing the anonymity requirement consists of several parts, one \emph{global} part and then, for each combination of cryptographers, a \emph{local} part. We start with the global part, which merely expresses that none of the individual bits that some cryptographer knows determines the value of any $p_i$ on its own, with the exception that if $p_i=1$, then of course cryptographer $c_i$ knows that $p_j=0$ for all $j\neq i$. The global part $\varphi_g$ is as follows:

\newcommand{\compatible}[2]{\ensuremath{\Diamond\Diamond(#1\wedge#2)}}
\newcommand{\independent}[2]{\ensuremath{\compatible{#1}{#2}\wedge\compatible{#1}{\overline{#2}}\wedge\compatible{\overline{#1}}{#2}\wedge\compatible{\overline{#1}}{\overline{#2}}}}

$$
\begin{array}{r}
\displaystyle\varphi_g=\bigwedge_{  \substack{         v\in\set{\bit_{\set{i,i+1}},\announce_i},  \\ k\in\set{0,\dots,n-1}}}
\hspace*{-1cm}\independent{v}{p_k}
\\
\wedge\bigwedge_{i\neq j}\compatible{p_i}{\overline{p_j}}\wedge\compatible{\overline{p_i}}{p_j}\wedge\compatible{\overline{p_i}}{\overline{p_j}}
\end{array}
$$

The first line of the formula requires that, for every variable $v$ of the $\bit_{\set{i,i+1}}$ or $\announce_i$-variables, and every cryptographer $c_k$, every combination of truth values of $v$ and $p_k$ appears. This encodes that the value of a \emph{single} variable $v$ does not give away any information about the  value of any $p_k$. The second line is a similar requirement for the value $p_i$: If $c_i$ paid, then she knows that $c_j$ did not pay, for $i\neq j$. However, the combination ``$p_i\wedge p_j$'' for $i\neq j$ should be the only one not appearing. Hence the formula requires that all other combinations appear in some final state. The global part $\varphi_g$ hence ensures that each individual bit that $c_i$ knows does not tell him whether $c_j$ paid, unless of course $i=j$ or $p_i=1$. 

The more interesting part is to encode that even the \emph{combination} of the above bits does not lead to additional knowledge; this is where the independence atom is crucial. We introduce some notation to enumerate the variables in $\pol{k}_i$:

\begin{itemize}
 \item for each $i$, let $\pol{k}_i=\set{v^i_1,\dots,v^i_{n+2}}$, 
 \item for $j\leq k$, let $V^i_{j\rightarrow k}=\set{v^i_j,\dots,v^i_k}$, 
 \item let $V^i_j=V^i_{j\rightarrow j}$. 
\end{itemize}

We now use modal independence logic to express that if each single variable from $\pol{k}_i$ does not tell $c_i$ anything about the value of $p_k$, then their combination does not, either. This is achieved with the following formula:

$$
\varphi^{i,k}=\Box\Box\left(  (V^i_1\bot_{p_k}V^i_2) \wedge (V^i_{1\rightarrow 2}\bot_{p_k} V^i_3) \wedge \dots \wedge (V^i_{1\rightarrow n+1}\bot_{p_k}V^i_{n+2}) \right).
$$

This formula requires that for each $j$, each pair of variable assignments $I_1$ to $V^i_{1\rightarrow j-1}$ and $I_2$ to $V^i_j$ that is ``locally compatible'' with some truth value $P(p_k)$---in other words, neither of these assignments by itself implies that the actual value of $p_k$ is not $P(p_k)$---is also compatible with that value for the combination of $I_1$ and $I_2$, i.e., there is some state satisfying $I_1\cup I_2\cup P$ (where the notion of a state satisfying a propositional assignment is defined as expected and the union of these assignments is well-defined since their domains are disjoint). As a 
consequence, the formula requires that for each $I\colon \pol{k}_i\rightarrow\set{0,1}$ and each $P\colon\set{p_k}\rightarrow\set{0,1}$, if for each $v\in \pol{k}_i$, there is a world $w$ such that $w\models \restr I{\set v}$ and $w\models P$, then there is a world $w$ such that $w\models I\cup P$. 

The following proposition formally states that our above-developed formulas indeed express the anonymity property of the protocol as intended. From a single cryptographer $c_i$'s point of view, it says that every observation $I$ which can arise when $c_i$ follows the protocol, as long as some cryptographer different from $c_i$ paid for the dinner, then for every $k$ different from $i$, both possibilities---$c_k$ paid for the dinner, or $c_k$ did not pay---cannot be ruled out by the observation $I$. We say that an assignment $I\colon\pol{k}_i\cup\set{\announce_i}\rightarrow\set{0,1}$ is \emph{consistent} if $i$ follows the protocol, i.e., if $I(\announce_i)=I(p_i)\oplus I(\bit_{\set{i,i-1}})\oplus I(\bit_{\set{i,i+1}})$. Note that in the models we are interested in, only consistent assignments appear in final states.

\begin{proposition}\label{prop:dc}
 If a Kripke model $M=(W,R,\Pi)$ satisfies the formula $\varphi_g\wedge\bigwedge_{i,k\in\set{0,\dots,n-1},i\neq k}\varphi^{i,k}$ at the world $q_0$, then the team $T=R(R(\set{q_0}))$ satisfies the following condition:
 For each $i\neq k\in\set{0,\dots,n-1}$ and each consistent $I\colon\pol{k}_i\cup\set{\announce_i}\rightarrow\set{0,1}$ with $I(p_i)=0$ and $\oplus_{j=0}^{n-1}I(\announce_j)=1$, there are worlds $w^I_1,w^I_2\in T$ with $w^I_1\models I,p_k$ and $w^I_2\models I,\overline{p_k}$.
\end{proposition}

We omit the easy proof; the proposition immediately follows from the semantics of the independence atom.

Our discussion only treats the anonymity property of the protocol. For a complete treatment, one also has to address other aspects as e.g., correctness, we omit this discussion here.
%


Note that in comparison to express the anonymity requirement using epistemic logic (see, e.g., \cite{AlBatainehvdMeyden-ABSTRACTION-MODEL-CHECKING-DINING-CRYPTOGRAPHERS-TARK-2011,Schnoor-Deciding-Epistemic-Crypt-Properties-ESORICS-2012}),
we do not use the relation of the Kripke model to represent knowledge, but to express branching time. In particular, our approach only uses a single modality.

\section{Expressiveness}
\label{sect:expressiveness}

We now compare the expressiveness of \MIL and classical modal logic, which we abbreviate with \ML. We show that \MIL is strictly more expressive than \ML on teams (simply because \MIL is not downwards closed), but that their expressiveness coincides on singleton teams. However, on singletons, \MIL is exponentially more succinct than \ML. We then study the expressiveness of \MIL with a generalized dependence atom as introduced in Section~\ref{sect:generalized dependence} instead of the independence atom. 

\subsection{Expressiveness of MIL and ML}

Clearly, since \MIL is not downward-closed, we obtain the following:

\begin{proposition}
 There is an \MIL-formula $\varphi_\MIL$ such that there is no \ML-formula $\varphi_\ML$ with the property that
$M,T\models\varphi_\MIL$ if and only if $M,T\models\varphi_\ML$
 for all models $M$ and all teams $T$. 
\end{proposition}

The proposition remains true%
\iflongproofs
\else\  
\fi
 for classical modal logic extended with a global modality, or for \MDL, since these logics remain downward-closed. In~\cite{se09}, it was shown that \MDL is as expressive as classical modal logic on singletons.
Therefore, a natural question to ask is whether on singletons, \MIL is still more expressive than \ML. We show that this is not the case, but we will also see that \MIL is exponentially more succinct than \ML, even on singletons.
For our proof, we use bisimulations, which are a well-established tool to compare expressiveness of different concepts. We recall the classical definition of bisimulation for modal logic:

\begin{definition}
 Let $M=(W,R,\Pi)$ and $M'=(W',R',\Pi')$ be Kripke models. A relation $Z\subseteq W\times W'$ is a \emph{modal bisimulation} if for every $(w,w')\in Z$, the following holds:
 \begin{itemize}
  \item $\Pi(w)=\Pi'(w')$, i.e., $w$ and $w'$ satisfy the same propositional variables,
  \item if $u$ is an $R$-successor of $w$, then there is an $R'$-successor $u'$ of $w'$ such that $(u,u')\in Z$ (forward condition),
  \item if $u'$ is an $R'$-successor of $w'$, then there is an $R$-successor $u$ of $w$ such that $(u,u')\in Z$ (backward condition).
 \end{itemize}
\end{definition}

It is well-known and easy to see that modal logic is invariant under bisimulation, i.e., if $Z$ is a bisimulation and $(w,w')\in Z$, then $w$ and $w'$ satisfy the same modal formulas. We now ``lift'' this property to modal independence logic by considering a bisimulation $Z$ as above on the team level:

\begin{definition}\label{Z-bisimilarity}
 Let $M=(W,R,\Pi)$ and $M'=(W',R',\Pi')$ be models, let $T\subseteq W$ and $T'\subseteq W'$ be teams. Let $Z\subseteq W\times W'$ be a modal bisimulation. Then $T$ and $T'$ are $Z$-bisimilar if the following is true:
 \begin{itemize}
  \item for each $w\in T$, there is a $w'\in T'$ such that $(w,w')\in Z$,
  \item for each $w'\in T'$, there is a $w\in T$ such that $(w,w')\in Z$.
 \end{itemize}
\end{definition}

We now show that on the team level, bisimulation for modal independence logic plays the same role as it does on the world level for modal logic: Simply stated, bisimilar teams satisfy the same formulas. Due to Lemma~\ref{lemma:dep_simulation}, the result also applies to modal dependence logic. This lemma may be of independent interest (for example, it implies a ``family-of-trees''-like model property), we use it to compare the expressiveness of \MIL and \ML.
  
\begin{lemma}\label{lemma:mil invariant under bisimulation}
 Let $M=(W,R,\Pi)$ and $M'=(W',R',\Pi')$ be Kripke models, let $T\subseteq W$ and $T'\subseteq W'$ be teams that are $Z$-bisimilar for a modal bisimulation $Z$. Then for any \MIL-formula $\varphi$, we have that 
$M,T\models\varphi$ if and only if $M',T'\models\varphi$.
\end{lemma}

\begin{proof}
 We show the lemma by induction on $\varphi$. Clearly it suffices to show that if $M,T\models_\MIL\varphi$, then $M',T'\models_\MIL\varphi$. Hence assume $M,T\models_\MIL\varphi$.
 \begin{itemize}
  \item Let $\varphi=x$ for some propositional variable $x$, and let $w'\in T'$. Since $T$ and $T'$ are $Z$-bisimilar, there is a world $w\in T$ with $(w,w')\in Z$. Since $M,T\models_\MIL x$, the variable $x$ is true at $w$ in $M$. Since $Z$ is a modal bisimulation, it follows that $x$ is true at $w'$ in $M'$, and hence every world $w'\in T$ satisfies $x$. Therefore, it follows that $M',T'\models\varphi$.
  \item If $\varphi=\neg x$, the proof is the same as above.
  \item Let $\varphi=\varphi_1\wedge\varphi_2$. This case trivially follows inductively.
  \item Let $\varphi=\varphi_1\vee\varphi_2$. Since $M,T\models\varphi$, it follows that $T=T_1\cup T_2$ for teams $T_1$ and $T_2$ with $M,T_1\models\varphi_1$ and $M,T_2\models\varphi_2$. We define teams $T_1'$ and $T_2'$ as follows:\
  \begin{itemize}
    \item $T_1'=\set{w'\in T'\ \vert\ (w,w')\in Z\mathtext{ for some }w\in T_1}$,
    \item $T_2'=\set{w'\in T'\ \vert\ (w,w')\in Z\mathtext{ for some }w\in T_2}$.
  \end{itemize}
  We prove the following:
  \begin{enumerate}
   \item $T'=T_1'\cup T_2'$
   \item $T_1$ and $T_1'$ are $Z$-bisimilar,
   \item $T_2$ and $T_2'$ are $Z$-bisimilar.
  \end{enumerate}
  By induction, it then follows that $M',T_1'\models\varphi_1$ and $M',T_2'\models\varphi_2$, which, since $T'=T_1'\cup T_2'$ implies that $M',T'\models\varphi$. We prove these points:
  \begin{enumerate}
   \item By construction, $T_1'\cup T_2'\subseteq T'$. Hence let $w'\in T'$. Since $T$ and $T'$ are $Z$-bisimilar, there is some $w\in T$ such that $(w,w')\in Z$. Since $T=T_1\cup T_2$, we can, without loss of generality, assume that $w\in T_1$. By definition of $T_1'$, it follows that $w'\in T_1'$.
   \item First let $w\in T_1\subseteq T$. Since $T$ and $T'$ are $Z$-bisimilar, there is some $w'\in T'$ such that $(w,w')\in Z$. Due to the definition of $T_1'$, it follows that $w'\in T_1'$. For the converse, assume that $w'\in T_1'$. By definition, there is some $w\in T_1$ such that $(w,w')\in Z$.
   \item This follows with the same proof as for $T_1$ and $T_1'$.
  \end{enumerate}
  Hence $M',T'\models\varphi$ as required.
  \item Let $\varphi=\Diamond\psi$. Since $M,T\models\varphi$, there exists a team $U\subseteq R(T)$ such that for each $w\in T$, the set $u(w):=R(\set{w})\cap T$ is not empty, and $M,U\models\psi$. We define a corresponding team $U'$ of $M'$ as follows: Start with $U'=\emptyset$ and then for each $(w,w')\in (T\times T')\cap Z$, do the following:
   \begin{itemize}
     \item For each $R$-successor $v$ of $w$ that is an element of $U$, since $Z$ is a modal bisimulation and $(w,w')\in Z$, there is at least one $R'$-successor $v'$ of $w'$ with $(v,v')\in Z$. Add all such $v'$ to the set $U'$.
   \end{itemize}
   
  By construction, $U'$ only contains worlds that are $R'$-successors of worlds in $T'$. Hence to show that $M',T'\models\Diamond\psi$, it remains to show that
  \begin{enumerate}
   \item for each $w'\in T'$, the team $U'$ contains a world $v'$ that is an $R'$-successor of $w'$,
   \item the teams $U$ and $U'$ are $Z$-bisimilar.
  \end{enumerate}
  
  The claim then follows by induction, since $M,U\models\psi$. We now show the above two points:
  \begin{enumerate}
   \item Let $w'\in T'$. Since $T$ and $T'$ are $Z$-bisimilar, there is some $w\in T$ with $(w,w')\in Z$. Due to the choice of $U$, there is some $v\in U$ which is an $R$-successor of $w$. By construction of the set $U'$, a world $v'$ that is an $R'$-successor of $w'$ has been added to $U'$.
   \item By construction, for every $R$-successor $v$ of some $w$ in $T$, at least one $v'$ has been added to $U'$ with $(v,v')\in Z$. For the converse, by construction of $U'$, for every $v'$ added to $U'$ there is a $v\in U$ with $(v,v')\in Z$.
  \end{enumerate}
  
  Hence $M',T'\models\Diamond\psi$ as required.
  \item Now assume that $\varphi=\Box\psi$, and let $U$ be the set of all $R$-successors of worlds in $T$, let $U'$ be the set of all $R'$-successors of worlds in $T'$. By induction, it suffices to show that $U$ and $U'$ are $Z$-bisimilar. Hence let $v\in U$ be the $R$-successor of some $w\in T$. Since $T$ and $T'$ are $Z$-bisimilar, there is some $w'\in T'$ such that $(w,w')\in Z$. Since $v$ is an $R$-successor of $w$, and since $Z$ is a modal bisimulation, there is some $v'$ which is an $R'$-successor of $w'$ such that $(v,v')\in Z$. Since $v'$ is an $R'$-successor of $w'$, it follows that $v'\in U'$. The converse direction follows analogously.
  \item Let $\varphi=\pol{p_1} \bot_{\pol{q}}\, \pol{p_2}$. To show that $M',T'\models\varphi$, let $u,u'\in T'$ with the same truth values of the variables in $q$. Since $T$ and $T'$ are bisimilar, there are worlds $w,w'\in T$ such that $(w,u)\in Z$ and $(w',u')\in Z$. Since $Z$ is a modal bisimulation, the $Z$-related worlds have the same propositional truth assignment. In particular, $w$ and $w'$ agree on the values for the variables in $q$. Since $M,T\models\varphi$, there is some world $w''\in T$ such that
  \begin{itemize}
    \item $w''\equiv_{\pol{q}}w'\equiv_{\pol{q}}w$, and since $Z$ is a modal bisimulation it follows that $w''$ and both $u$ and $u'$ have the same $\pol{q}$-assignment,
    \item $w''\equiv_{\pol{p_1}}w$, and hence $w''$ and $u$ have the same $\pol{p_1}$-assignment,
    \item $w''\equiv_{\pol{p_2}}w'$, and hence $w''$ and $u'$ have the same $\pol{p_2}$-assignment.
  \end{itemize}
  Since $T$ and $T'$ are $Z$-bisimilar, there is a world $u''\in T'$ such that $(w'',u'')\in Z$. Since $Z$ is a modal bisimulation, $w''$ and $u''$ satisfy the same propositional variables, and hence for $u''$ we have that
  \begin{itemize}
    \item $u''\equiv_{\pol q} u$, $u''\equiv_{\pol q} u'$
    \item $u''\equiv_{\pol{p_1}} u$
    \item $u''\equiv_{\pol{p_2}} u'$.
  \end{itemize}
  Therefore, $M',T'\models\varphi$ as required.
 \end{itemize}
\end{proof}

With Lemma~\ref{lemma:mil invariant under bisimulation} and an application of van Benthem's Theorem \cite{BenthemBook}, it follows directly that \MIL and \ML are in fact equivalent in expressiveness \emph{on singletons}:

\begin{theorem}\label{theorem:singleton mil expressiveness}
 For each \MIL-formula $\varphi_\MIL$, there is an \ML-formula $\varphi_\ML$ such that $\varphi_\ML$ and $\varphi_\MIL$ are equivalent on singletons.
\end{theorem}
\begin{proof}
	 Due to Lemma~\ref{lemma:mil invariant under bisimulation}, we know that \MIL is invariant under bisimulation of teams. Since for singleton teams, bisimulation on teams and bisimulation on worlds coincide, it follows that \MIL on singletons is invariant under modal bisimulation. Clearly, when evaluating an \MIL-formula $\varphi_\MIL$ on a singleton team $\set w$, all worlds in the model that have a distance from $w$ which exceeds the modal depth (i.e., maximal nesting degree of modal operators) of $\varphi$, are irrelevant for the question whether $M,\set{w}\models\varphi$. Therefore, $\varphi_\MIL$, evaluated on singletons, captures a property of Kripke models that is invariant under modal bisimulation and only depends on the worlds that can be reached in at most $\md{\varphi_\MIL}$ steps. Due to~\cite{vanbenthPHD}, such a property can be encoded by a standard modal logic formula $\varphi_\ML$. (The formula $\varphi_\ML$ can be obtained, for example, as the disjunction of formulas $\varphi_{M,w}$ which, for each model $M$ and world $w$ with $M,\set{w}\models_\MIL\varphi_\MIL$, encodes the finite tree unfolding of $M,w$ up to depth $\md{\varphi_\MIL}$, up to bisimulation, and only taking into account the variables appearing in $\varphi_\MIL$. This unfolding is finite and hence first-order definable, therefore we can apply the result from~\cite{vanbenthPHD}.)
\end{proof}

Since the application of van Benthem's Theorem yields a potentially very large formula, the above result does not give a ``efficient'' translation from \MIL to \ML. It turns out that one cannot do much better: \MIL is exponentially more succinct than \ML.

\begin{theorem}\label{theorem:succinctness}
 There is a family of \MIL-formulas $(\varphi_i)_{i\in\mathbb N}$ such that the length of $\varphi_i$ is quadratic in $i$, and for any family of \ML-formulas $(\psi_i)_{i\in\mathbb N}$ such that for all $i$, $\varphi_i$ and $\psi_i$ are equivalent on singletons, the length of $\psi_i$ grows exponentially in $i$.
\end{theorem}
\begin{proof}
	 Let $\varphi_i$ be the formula describing the security property of the dining cryptographers protocol in Section~\ref{sect:dining cryptographers}, with every sequence $\Diamond\Diamond$ replaced with $\Diamond$. As argued in that section, if $\varphi_i$ is satisfied at $M,\set{w}$, then the number of propositional assignments appearing in the set of worlds that can be reached from $w$ in one step (two steps for the original formula) is exponential in $i$. Now let $(\psi_i)_{i\in\mathbb N}$ a family of \ML-formulas such that $\varphi_i$ and $\psi_i$ are locally equivalent for each $i$. 
 Since $\psi_i$ is locally equivalent to $\varphi_i$, we can without loss of generality assume that $\md{\psi_i}=1$ for all $i$ (if $\psi_i$ contains deeper nestings of modal operators, the formula can be simplified since the truth value of $\psi_i$ cannot depend on worlds reachable in $2$ or more steps). Therefore, modal operators do not appear nested in $\psi_i$. Without loss of generality, we can assume that only $\Diamond$ appears in $\psi$. It is clear that if $M,w\models\psi_i$, then there is a submodel of $M$ which contains $w$, and in which the number of successors of $w$ is bounded by the number of $\Diamond$-operators appearing in $\psi_i$. Therefore, $\psi_i$ must have an exponential number of $\Diamond$-operators, which proves the theorem.
\end{proof}


\subsection{Expressiveness of \ML with Generalized Dependence Atoms}

Theorem~\ref{theorem:singleton mil expressiveness} applies not only to \MIL, but (with the same proof) to all extensions of \MIL with a dependence operator that, evaluated on a team $T$, depends only on the set of propositional assignments that occur in some world of the team. This is because the set of assignments is clearly invariant under bisimulation. As examples, operators like -those from exclusion logic or inclusion logic can be added to \MIL without increasing its expressiveness. Hence, in light of Section~\ref{sect:generalized dependence}, a natural question to ask is the following: For which generalized dependence atoms $D$ is the logic \MILD as expressive as \ML on singletons, and which atoms do in fact add expressiveness?

--It is easy to see that there are generalized dependence atoms $D$ that, even on singletons, add expressiveness beyond classical modal logic (and hence, beyond \MIL). This is because with no restriction on the dependence operator $D$, one can express properties that depend on the \emph{number} of worlds in a team, which clearly cannot be done in ML. As an example, consider the following:

\medskip

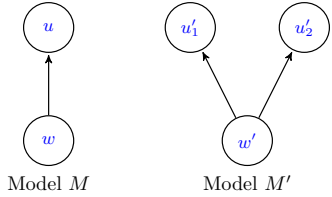
\scalebox{0.75}{
\begin{minipage}{1.5cm}
\begin{center}
 \begin{tikzpicture}[->,>=stealth',shorten >=1pt,auto,node distance=1.75cm, semithick]
   \tikzstyle{every state}=[fill=white,draw=black,text=blue]
   \node[state] at (0,0)   (A) {\begin{small}$w$\end{small}};
   \node[state] at (0,2)   (B) {\begin{small}$u$\end{small}};
   \draw [->] (A) edge (B);
\end{tikzpicture}
Model $M$
\end{center}
\end{minipage}
\hspace*{1cm}
\begin{minipage}{3cm}
\begin{center}
\begin{tikzpicture}[->,>=stealth',shorten >=1pt,auto,node distance=1.75cm, semithick]
   \tikzstyle{every state}=[fill=white,draw=black,text=blue]
   \node[state] at (0,0)   (A) {\begin{small}$w'$\end{small}};
   \node[state] at (-1,2)  (B) {\begin{small}$u_1'$\end{small}};
   \node[state] at (1,2)   (C) {\begin{small}$u_2'$\end{small}};
   \draw [->] (A) edge (B);
   \draw [->] (A) edge (C);
\end{tikzpicture}
Model $M'$
\end{center}
\end{minipage}}

\medskip

\begin{example}
 Let $D$ be the relation $\{(0)\}$, and consider the models $M$ and $M'$ above, where the single variable $x$ is false in every world of both models.
 It is easy to see that $M,\set w\models\Box D(x)$, while $M',\set{w'}\not\models\Box D(x)$: For the model $M$, the set of successors of $w$ is the team $T=\set u$, since $x$ is false in $u$, it follows that $T(x)=(0)$, hence $T\models D(x)$. On the other hand, the set of successors of $w'$ in $M'$ is $T'=\set{u_1,u_2}$, hence $T'(x)=(0,0)$ and $T'\not\models D(x)$.
 
 Clearly, no \ML-formula can distinguish $M,w$ and $M',w'$, since the relation $Z=\set{(w,w'),(u,u_1),(u,u_2)}$ is a bisimulation.
\end{example}
%
%

However, as mentioned earlier, the proof of Theorem~\ref{theorem:singleton mil expressiveness} can be generalized to handle the generalized dependence atoms discussed earlier. In general, we obtain the following result: For a $D$ which is first-order definable, the expressiveness of \MILD coincides with that on \ML for singletons. For the proof, we show that \MILD remains invariant under bisimulation, and then apply the proof of  Theorem~\ref{theorem:singleton mil expressiveness} again. On the other hand, clearly if $D$ is not FO-definable, then $D$ cannot be expressed in modal logic, as modal logic can be translated to first-order logic with the standard translation. Hence we obtain the following theorem:

\begin{theorem}\label{theorem:mil+d expressiveness singletons}
 Let $D$ be a generalized dependence atom. Then the following statements are equivalent:
 \begin{enumerate}
  \item $D$ can be expressed in first-order logic,
  \item \MILD and \ML are equally expressive over singletons, i.e., for each \MILD-formula $\varphi_\MIL$, there is an \ML-formula $\varphi_\ML$ such that $\varphi_\MIL$ and $\varphi_\ML$ are equivalent on singletons.
 \end{enumerate}
\end{theorem}
\begin{proof}
	 We first assume that $D$ is FO-definable.
 As mentioned above, it suffices to adapt the proof of Lemma~\ref{lemma:mil invariant under bisimulation} to \MILD. Clearly, in the induction, we only need to cover the case that $\varphi=D(p_1,\dots,p_n)$ for propositional variables $p_1,\dots,p_n$. Hence assume that $M,T\models D(p_1,\dots,p_n)$, we show that $M',T'\models D(p_1,\dots,p_n)$, where $T$ and $T'$ are $Z$-bisimilar for a modal bisimulation $Z$.
 
 Let $\phi$ be the first-order formula defining $D$. We prove the claim inductively over the formula, where we only cover the key cases explicitly. Let the free variables of $\phi$ be $\omega_1,\dots,\omega_t$. We show that if $w_1,\dots,w_t\in T$ and $w_1',\dots,w_n'\in T'$ such that $(w_i,w_i')\in Z$ for all $i$, then $\phi(w_1,\dots,w_n)$ evaluates to true if and only if $\phi(w_1',\dots,w_n')$ does. If $\phi$ is quantifier-free, the claim is clear: Since $Z$ is a modal bisimulation, $w_i$ and $w_i'$ satisfy the same propositional variables, and due to the choice of $\phi$, the truth value of $\phi$ only depends on the propositional assignments of the worlds instantiating the variables $\omega_1,\dots,\omega_n$. The second relevant case is when $\phi=\exists \omega\psi(\omega,\omega_1,\dots,\omega_t)$. If $M,T\models \phi$, then there is a world $w\in T$ such that $\psi(w,w_1,\dots,w_t)$ is true. Since $T$ and $T'$ are $Z$-bisimilar, it follows that there is a world $w'\in T'$ such that $(w,w')\in Z$. Due to induction, it follows that if $\psi(w,w_1,\dots,w_t)$ is true, then so is $\psi(w',w_1',\dots,w_t')$. This completes the proof that \MILD and ML are equally expressive over singletons.
 
 For the converse, assume that \MILD is as expressive as ML over singletons. In particular, then for every sequence $x_1,\dots,x_n$ of variables, there is a modal formula $\varphi$ such that for every model $M$ and every world $w\in M$, we have that $M,w\models\varphi$ if and only if $M,\set{w}\models\Box D(x_1,\dots,x_n)$. By the standard translation from modal logic to first-order logic, this implies that $D(x_1,\dots,x_n)$ can be expresses as an FO-formula.
\end{proof}

\section{Conclusion and Open questions}

In this paper we introduced modal independence logic \MIL and settled the computational complexity of its satisfiability and model checking problem. Furthermore we compared the expressivity of \MIL with that of classical modal logic. It turned out that most of our results can be generalized to modal logic extended with various so called generalized dependence atoms. 

We end this paper by the following interesting open questions:
\begin{enumerate}
\item Are there classes of frames definable in \MIL that cannot be defined with \ML?
\item Is it possible to formulate and prove a version of  van Benthem's Theorem for our generalization of bisimilarity from the world level to the team level (see Definition \ref{Z-bisimilarity})?
\end{enumerate}

\bibliographystyle{aiml14}
\bibliography{mil}

\end{document}